\documentclass[sigconf]{acmart}

\usepackage{booktabs} 
\usepackage{balance}
\usepackage{itrans}
\usepackage{subcaption}
\newtheorem{remark}{Remark}

\begin{document}
	\copyrightyear{2017}
	\setcopyright{acmcopyright}
	\acmConference{SocialSens'2017}{Apr. 21 2017}{Pittsburgh, PA, USA} 
	\acmISBN{978-1-4503-4977-2/17/04}
	\acmPrice{15.00}
	\acmDOI{{http://dx.doi.org/10.1145/3055601.3055607}}
	
\title[Does Confidence Reporting from the Crowd Benefit Crowdsourcing]{Does Confidence Reporting from the Crowd Benefit Crowdsourcing Performance?}

\author{Qunwei Li}
\affiliation{%
  \institution{Department of EECS, Syracuse University}
  \state{NY} 
  \postcode{13210}
}
\email{qli33@syr.edu}

\author{Pramod K. Varshney}
\affiliation{%
	\institution{Department of EECS, Syracuse University}
	\state{NY} 
	\postcode{13210}
}
\email{varshney@syr.edu}

\renewcommand{\shortauthors}{Qunwei Li, et al.}

\begin{abstract}
We explore the design of an effective crowdsourcing system for an $M$-ary classification task. Crowd workers
complete simple binary microtasks whose results are aggregated to give the final classification decision. We consider
the scenario where the workers have a reject option so that they are allowed to skip microtasks when they are unable to or choose not to respond to binary microtasks. Additionally, the workers report quantized confidence levels when they are able to submit definitive answers. We present an aggregation approach using a weighted majority voting rule, where each worker's response is assigned an optimized weight to maximize crowd's classification performance. We obtain a couterintuitive result that the classification performance does not benefit from workers reporting quantized confidence. Therefore, the crowdsourcing system designer should employ the reject option without requiring confidence reporting.
\end{abstract}

%
%

\ccsdesc[500]{Human-centered computing~Social network analysis}


\keywords{Classification, crowdsourcing, distributed inference, information fusion, reject option, confidence reporting}

\maketitle

\section{Introduction}
{Crowdsourcing}  provides a new framework to utilize distributed human wisdom to solve problems that machines cannot perform well, like handwriting recognition, paraphrase acquisition, audio transcription, and photo tagging \cite{paritosh2011computer,burrows2013paraphrase,7052378}. Despite the successful applications of crowdsourcing, the relatively low quality of output is a key challenge \cite{IpeirotisPW2010,allahbakhsh2013quality,mo2013cross}.

Several methods have been proposed to deal with the aforementioned problems \cite{KargerOS2011b,VempatyVV2014,yue2014weighted,6891807,VarshneyVV2014,QuinnB2011,zhang2012reputation,hirth2013analyzing}. A crowdsourcing task is decomposed into microtasks that are easy for an individual to accomplish, and these microtasks could be as simple as binary distinctions \cite{KargerOS2011b}. A classification problem with crowdsourcing, where taxonomy and dichotomous keys are used to design binary questions, is considered in \cite{VempatyVV2014}. In our research group, we employed binary questions and studied the use of error-control codes and decoding algorithms to design crowdsourcing systems for reliable classification \cite{VarshneyVV2014,VempatyVV2014}.  A group control mechanism where the reputation of the workers is taken into consideration to partition the crowd into groups is presented in\cite{QuinnB2011,zhang2012reputation}. Group control and majority voting are compared in \cite{hirth2013analyzing}, which reports that majority voting is more cost-effective on less complex tasks.

In past work on classification via crowdsourcing, crowd workers were required to provide a definitive yes/no response to binary microtasks. Crowd workers may be unable to answer questions for a variety of reasons such as lack of expertise. As an example, in mismatched speech transcription, i.e., transcription by workers who do not know the language, workers may not be able to perceive the phonological dimensions they are tasked to differentiate \cite{jyothi2015acquiring}. In recent work, we have investigated the design of the optimal aggregation rule when the workers have a reject option so that they are unable to or choose not to respond \cite{7747496}.

The possibility of using confidence scores to improve the quality of crowdsourced labels was investigated in \cite{kazai2011search}. An aggregation method using confidence scores to integrate labels provided by crowdsourcing workers was developed in \cite{oyama2013accurate}. A payment mechanism was proposed for crowdsourcing systems with a reject option and confidence score reporting \cite{shah2014double}. Indeed, confidence reporting can be useful for estimating the quality of the provided responses and possibly yield better outcomes when the aggregation is not optimal. However, potential crowdsourcing performance improvement with an optimal aggregation rule resulting from confidence reporting has not yet been investigated. As is studied in this paper, when an optimal aggregation rule is developed, confidence reporting does not help to improve the performance. 

In this paper, we further consider the problem investigated in \cite{7747496} by studying the scenario when the workers include their confidence levels in their responses.
The main contribution of this paper is the counterintuitive finding that the confidence scores of the crowd do not play a role in the optimal aggregation rule. The weight assignment scheme to ensure the maximum weight for the correct class is the same as that when there is no confidence reporting. Although confidence reporting can provide useful information for estimating the quality of the crowd, the noise introduced due to categorization of confidence makes the estimation less accurate. Since the estimation result is essential for aggregation, confidence reporting may cause performance degradation.

\section{Crowdsourcing Task with a Reject Option}
Consider the situation where $W$ workers take part in an $M$-ary object classification task. Each worker is asked $N$ simple binary questions, termed as microtasks, and the worker's answer to a single microtask is conventionally represented by either ``1'' (Yes) or ``0'' (No), which eventually lead to a classification decision among the $M$ classes. We assume independent microtask design and, therefore, we have $N = \left\lceil {{{\log }_2}M} \right\rceil $ independent microtasks of equal difficulty. The workers submit responses that are combined to give the final decision. Here, we consider the microtasks to be simple binary questions and the worker's answer to a single microtask is conventionally represented by either ``1'' (Yes) or ``0'' (No) \cite{VempatyVV2014,rocker2007paper}. Thus, the $w$th worker's ordered answers to all the microtasks form an $N$-bit word, which is denoted by ${\bf a}_w$. Let ${\bf a}_w(i)$, $i\in \{ 1,2,\dots,N\}$ represent the $i$th bit in this vector.  

In our previous work \cite{7747496}, we considered a general problem setting where the worker has a reject option of skipping the microtasks. We denote this skipped answer as $\lambda$, whereas the ``1/0'' (Yes/No) answers are termed as definitive answers. Due to the variability of different worker backgrounds, the probability of submitting definitive answers is different for different workers. Let $p_{w,i}$ represent the probability of the $w$th worker submitting $\lambda$ for the $i$th microtask. Similarly, let $\rho_{w,i}$ be the probability that ${\bf a}_w(i)$, the $i$th answer of the $w$th worker, is correct given that a definitive answer has been submitted. Due to the variabilities and anonymity of workers, we study crowdsourcing performance when $p_{w,i}$ and $\rho_{w,i}$ are realizations of certain probability distributions, which are denoted by distributions ${F_P }\left( p \right)$ and ${F_\rho }\left( \rho \right)$ respectively. The corresponding means are expressed as $m$ and $\mu$.

Let $H_0$ and $H_1$ denote the hypotheses where ``0'' or ``1'' is the true answer for a single microtask, respectively. For simplicity of performance analysis, $H_0$ and $H_1$ are assumed equiprobable for every microtask. The crowdsourcing task manager or a fusion center (FC) collects the $N$-bit words from $W$ workers and performs fusion based on an aggregation rule. 
We focus on finding the optimal aggregation rule and let us briefly review the results regarding the aggregation of responses from the workers for classification in our previous work \cite{7747496}.
\\
\textbullet\ Let $D=\{e_j, j=1,2,\dots,M\}$ be the set of all the object classes, where $e_j$ represents the $j$th class. Based on $w$th worker's response to the microtasks, a subset $D_w$ is chosen, within which the classes are associated with weight $W_w$ for aggregation.\footnote{If all the responses from the $w$th worker are definitive, $D_w$ is a singleton. Otherwise, $D_w$ contains multiple classes.}
	The fusion center FC adds up the weights for every class and chooses the one with highest overall weight as the final decision $e_D$, which can be expressed as
	\begin{align}
	e_D= \arg \mathop {\max }\limits_{{{e_j} \in {D}} }\left\{  \sum\limits_{w= 1}^W {{W_w}I_{{D_w}}\left\langle {{e_j}} \right\rangle } \right\}, j=1,2,\dots ,M,
	\end{align}
	where $I_{{D_w}}\left\langle {{e_j}} \right\rangle$ is an indicator function which equals 1 if $e_j\in D_w$ and 0 otherwise. To derive the optimal weight $W_w$ for each worker, one may look into the minimization of the misclassification probability, for which a closed-form expression cannot be derived without an explicit expression for $W_w$. Hence, it is difficult to determine the optimal weight.
\\	
	\textbullet\  The $M$-ary classification task can also be split into $N$ binary hypothesis testing problems, by associating a classification decision with an $N$-bit word. Each worker votes ``1'' or ``0'' with the weight $W_w$ for every bit. In this case, the Chair-Varshney rule gives the optimal weight as $W_w={\log \frac{{{\rho _{w,i}}}}{{1 - {\rho _{w,i}}}}}$ \cite{ChairV1986}. However, this requires the prior knowledge on $\rho_{w,i}$ for every worker, which is not available in practice. 
\\	
	\textbullet\  We proposed a novel weighted majority voting method, which was derived by solving the following optimization problem
	\begin{equation}\label{problem}
	\begin{array}{l}
	\text{maximize}\ \ {E_C}\left[ {{\mathbb{W}}} \right]\\
	\text{subject to}\ \ {E_O}\left[ {{\mathbb{W}}} \right] = {K}
	\end{array}
	\end{equation}
	where ${E_C}\left[ {{\mathbb{W}}} \right]$ denotes the crowd's average weight contribution to the correct class and ${E_O}\left[ {{\mathbb{W}}} \right]$ denotes the average weight contribution to all the possible classes that is constrained to remain a constant $K$. Statistically, this method ensures maximum weight to the correct class and consequently maximum probability of correct classification. We showed that this method significantly outperforms the simple majority voting procedure.

In this paper, we investigate the impact of confidence reporting from the crowd on system performance. The weight assignment scheme is developed by solving problem \eqref{problem} as well.

\section{Crowdcouring with Confidence Reporting}
We consider the case where the crowd is composed of honest workers, which means that the workers honestly observe, think, and answer the questions, give confidence levels, and skip questions that they are not confident about. We derive the optimal weight assignment for the workers and the performance of the system in a closed form. Based on these findings, we determine the potential benefits of confidence reporting in a crowdsourcing system with a reject option.

\subsection{Confidence Level Reporting}
In a crowdsourcing system where workers submit answers and report confidence, 
we define the $w$th worker's confidence about the answer to the $i$th microtask as the probability of this answer being correct given that this worker gives a definitive answer, which is equal to $\rho_{w,i}$ as defined earlier. When $\rho_{w,i}$ is bounded as $\frac{ l_{w,i}-1}{L}\le \rho_{w,i} \le \frac{ l_{w,i}}{L}$, $l_{w,i}\in \{1,\ldots,L\}$, the $w$th worker reports his/her confidence level as $l_{w,i}$. Let $l_{w,i}$ be drawn from the distribution $l_{w,i} \sim F_L(l)$. Note that every worker independently gives confidence levels for different microtasks, and $L=1$ simply means that workers submit answers and do not report their confidence levels. 

Assuming that a worker can accurately perceive the probability $\rho_{w,i}$ and honestly report the confidence level, intuitively it is expected that it will benefit the crowdsourcing fusion center as much more information about the quality of the crowd can be extracted. However, as the confidence is quantized, which helps the workers in determining the confidence levels to be reported, quantization noise is introduced in extracting the crowd quality from confidence reporting.

As an illustrative example, consider the problem of mismatched crowdsourcing for speech transcription, which has garnered interest in the signal processing community \cite{Hasegawa-JohnsonCJV2015,jyothi2015acquiring,VarshneyJH2016,LiuJTMSKHK2016,ChenHC2016,KongJH2016}.
Suppose the four possibilities for a velar stop consonant to transcribe are $R=\{${\fransdvng k}, {\fransdvng K}, {\fransdvng g}, {\fransdvng G}$\}$. The simple binary question of ``whether it is aspirated or unaspirated'' differentiates between $\{${\fransdvng K}, {\fransdvng G}$\}$ and $\{${\fransdvng k}, {\fransdvng g}$\}$, whereas the binary question of ``whether it is voice or unvoiced'' differentiates between $\{${\fransdvng g}, {\fransdvng G}$\}$ and $\{${\fransdvng k}, {\fransdvng K} $\}$. The highest confidence level is set as $L=4$. Now suppose the first worker is a native Italian speaker.  Since Italian does not use aspiration, this worker will be unable to differentiate between $\{${\fransdvng k}$\}$ and $\{${\fransdvng K}$\}$, or between $\{${\fransdvng g}$\}$ and $\{${\fransdvng G}$\}$.  It would be of benefit if this worker would specify the inability to perform the task through a special symbol $\lambda$, rather than guessing randomly, and this worker answers ``Yes'' with confidence level 1 to the second question. Suppose the second worker is a native Bengali speaker. Since this language makes a four-way distinction among velar stops, such a worker will probably answer both questions without a $\lambda$.  

In the rest of this section, we address the problem ``Does the confidence reporting help crowdsourcing system performance?'' by performing analyses when workers report their confidences with their definitive answers. 
\subsection{Optimal Weight Assignment Scheme}
We determine the optimal weight $W_w$ for the $w$th worker in this section. We rewrite hereby the weight assignment problem 
\begin{equation}\label{max}
\begin{array}{l}
\text{maximize}\ \ {E_C}\left[ {{\mathbb{W}}} \right]\\
\text{subject to}\ \ {E_O}\left[ {{\mathbb{W}}} \right] = {K}
\end{array}
\end{equation}
where ${E_C}\left[ {{\mathbb{W}}} \right]$ denotes the crowd's average weight contribution to the correct class and ${E_O}\left[ {{\mathbb{W}}} \right]$ denotes the average weight contribution to all the possible classes and remains a constant $K$. Statistically, we are looking for the weight assignment scheme such that the weight contribution to the correct class is maximized while the weight contribution to all the classes remains fixed, so as to maximize the probability of correct classification.
\begin{proposition}\label{pro1}
	To maximize the average weight assigned to the correct classification element, the weight for $w$th worker's answer is given by
	\begin{align}
	W_w={\mu}^{-n},
	\end{align}
	where $n$ is the number of definitive answers that the $w$th worker submits.
\end{proposition}

\begin{proof}	
	See Appendix.
\end{proof}
\begin{remark}
	Here the weight depends on the number of questions answered by a worker. In fact, if more questions are answered, the weight assigned to the corresponding worker's answer is larger. This is intuitively pleasing as a high-quality worker is able to answer more questions and is assigned a higher weight. Increased weight can put more emphasis on the contribution of high-quality workers in that sense and improve overall classification performance.
\end{remark}

\begin{remark}
	When $L=\infty$, $\rho_{w,i}$ associated with every worker for every microtask is reported exactly. Then the Chair-Varshney rule gives the optimal weight assignment to minimize error probability \cite{ChairV1986}. However, human decision makers are limited in their information processing capacity and can only carry around
	seven categories \cite{miller1956magical}. Thus, the largest value of $L$ is around 7 in practice.
\end{remark}
\begin{remark}
	Note that the optimal weight assignment scheme is the same as in the case where the workers do not report confidence levels, i.e., $L=1$. Actually, the value of $L$ does not play any role in the weight assignment, as long as $\rho_{w,i}$ is not known exactly. Therefore, the weight assignment is universally optimal regardless of confidence reporting. 
\end{remark}
\subsection{Parameter Estimation}
Before the proposed aggregation rule can be used, $\mu$ has to be estimated to assign the weight for every worker's answers. Here, we employ three approaches to estimate $\mu$. We refer to our previous work \cite{7747496} for training-based and majority-voting based methods to estimate $\mu$, and give an additional method using the information extracted from the workers' reported confidence levels.
\subsubsection*{Confidence-based}Note that the reported confidence levels correspond to $\rho_{w,i}$. We collect all the values of the submitted confidence levels and obtain the estimate of $\mu$ from them. First, the $w$th worker's confidence level for the $i$th microtask is represented by $l_{w,i}$. Considering the fact that $\frac{ l_{w,i}-1}{L}\le \rho_{w,i} \le \frac{ l_{w,i}}{L}$ if the worker submits a definitive answer, we use $\frac{ l_{w,i}-\frac 1 2}{L}$ to approximate $\rho_{w,i}$. Let $l_{w,i}=\frac 1 2$ if the $w$th worker skips the $i$th microtask.  We obtain the estimate of $\mu$ by
\begin{align}
\hat \mu=\frac{1}{W-\epsilon}\sum\limits_{w=1}^{W}\sum\limits_{i=1}^N \frac{ l_{w,i}-\frac 1 2}{LI(w)},
\end{align}
where $I(w)$ denotes the number of definitive answers that $w$th worker submits.

\subsection{Performance Analysis}
In this section, we characterize the performance of the proposed crowdsourcing classification framework in terms of the probability of correct classification $P_c$. Note that we have overall correct classification only when all the bits are classified correctly.

\begin{proposition}\label{pro2}
	The probability of correct classification $P_c$ in the crowdsourcing system is
	\begin{align}
	P_c=\Big[\frac{1}{2}&+ \frac{1}{2}\sum\limits_S {\binom{W}{\mathbb{Q}}} \left( {F\left( \mathbb{Q} \right) - F^{\prime}\left( \mathbb{Q} \right)} \right)\nonumber \\
	&+ \frac{1}{4}\sum\limits_{S^\prime} {\binom{W}{\mathbb{Q}}} \left( {F\left( \mathbb{Q} \right) - F^{\prime} \left( \mathbb{Q} \right)} \right)\Big]^N,
	\end{align}
	
	where
	$
	\mathbb{Q}=\left\{({{q_{ - {N}}},{q_{ - {{N +1}}}}, \ldots {q_{{N}}}}):  \sum\limits_{n = -N}^N{{q_{{n}}}  = W}\right\}\nonumber
	$ with natural numbers $q_n$ and $q_0$, and
	$
	{S} = \left\{ {\mathbb{Q}: } {\sum\limits_{n = 1}^N {{{\mu}^{-n}}\left( {{q_{{n}}} - {q_{ - {n}}}} \right)}  > 0}\right\}\nonumber,
	$
	$
	S^\prime = \left\{ {\mathbb{Q}:\sum\limits_{n = 1}^N {{\mu ^{ - n}}\left( {{q_n} - {q_{ - n}}} \right)}  = 0} \right\}\nonumber,
	$ $\binom{W}{\mathbb{Q}} = \frac{{W!}}{{\prod\limits_{n =  - N}^N {{q_n}!} }}$, and
	\begin{align}
	F({\mathbb{Q}}) = {m^{{q_0}}}\prod\limits_{n = 1}^N {{{\left( {1 - \mu } \right)}^{{q_{ - n}}}}{\mu ^{{q_n}}}{{\left( {C_{N - 1}^{n - 1}{{\left( {1 - m} \right)}^n}{m^{N - n}}} \right)}^{{q_{ - n}} + {q_n}}}} \nonumber
	\end{align}
	\begin{align}
	F^{\prime}({\mathbb{Q}}) = {m^{{q_0}}}\prod\limits_{n = 1}^N {{{\left( {1 - \mu } \right)}^{{q_n}}}{\mu ^{{q_{ - n}}}}{{\left( {C_{N - 1}^{n - 1}{{\left( {1 - m} \right)}^n}{m^{N - n}}} \right)}^{{q_{ - n}} + {q_n}}}}\nonumber .
	\end{align}
\end{proposition}
\begin{proof}
	The proof is similar to the proof in our previous work \cite{7747496} and is, therefore, omitted for brevity.
\end{proof}

\section{Simulation Results}
In this section, we give the simulation results for the proposed crowdsourcing system. The workers take part in a classification task of $N=3$ microtasks. $F_P(p)$ is a uniform distribution denoted as $U(0,1)$.

First, we show the efficiency of the derived optimal weight assignment over the widely used simple majority voting method for crowdsourcing systems. The performance comparison is presented with the number of workers varying from 3 to 29. Here, we consider different qualities of the individual workers in the crowd which is represented by variable $\rho_{w,i}$ with a uniform distribution $U(0.6,1)$. Thus, the mean $\mu$ is 0.8, and we give simulation results when confidence reporting is not included and the estimation of $\mu$ is perfect in Fig. \ref{variousWorkerSize}. It is observed that a larger crowd completes the classification task with higher quality. A significant performance improvement by the proposed method with a reject option compared with the simple majority voting is shown in the figure.

\begin{figure}[h]
	\centering
	\includegraphics[width=3in]{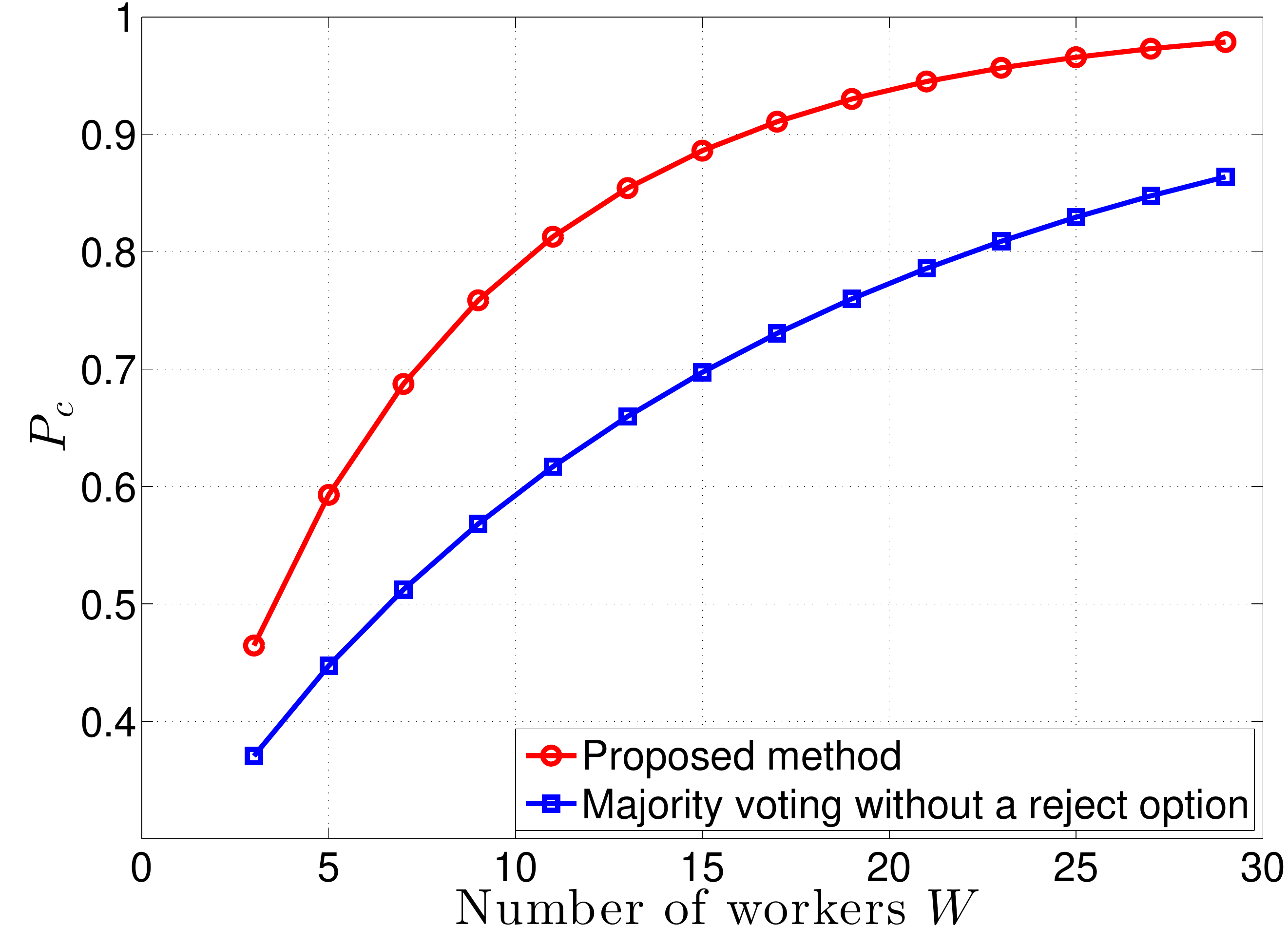} 
	\caption{Performance comparison with various crowd sizes.}
	\label{variousWorkerSize}
\end{figure}

Since an accurate estimation of $\mu$ is essential for applying the optimal weight assignment scheme, we next focus on the estimation results of $\mu$ for the three estimation methods as discussed in the previous section. Let $F_\rho(\rho)$ be a uniform distribution expressed as $U(x,1)$ with $0\le x \le 1$, and thus we can have $\mu$ varying from 0.5 to 1. We consider that $W=20$ workers participate in the classification task with a reject option and confidence reporting.
\begin{figure}[htp]
	\centering
	\begin{minipage}[t]{0.5\linewidth}
		\includegraphics[width=\linewidth]{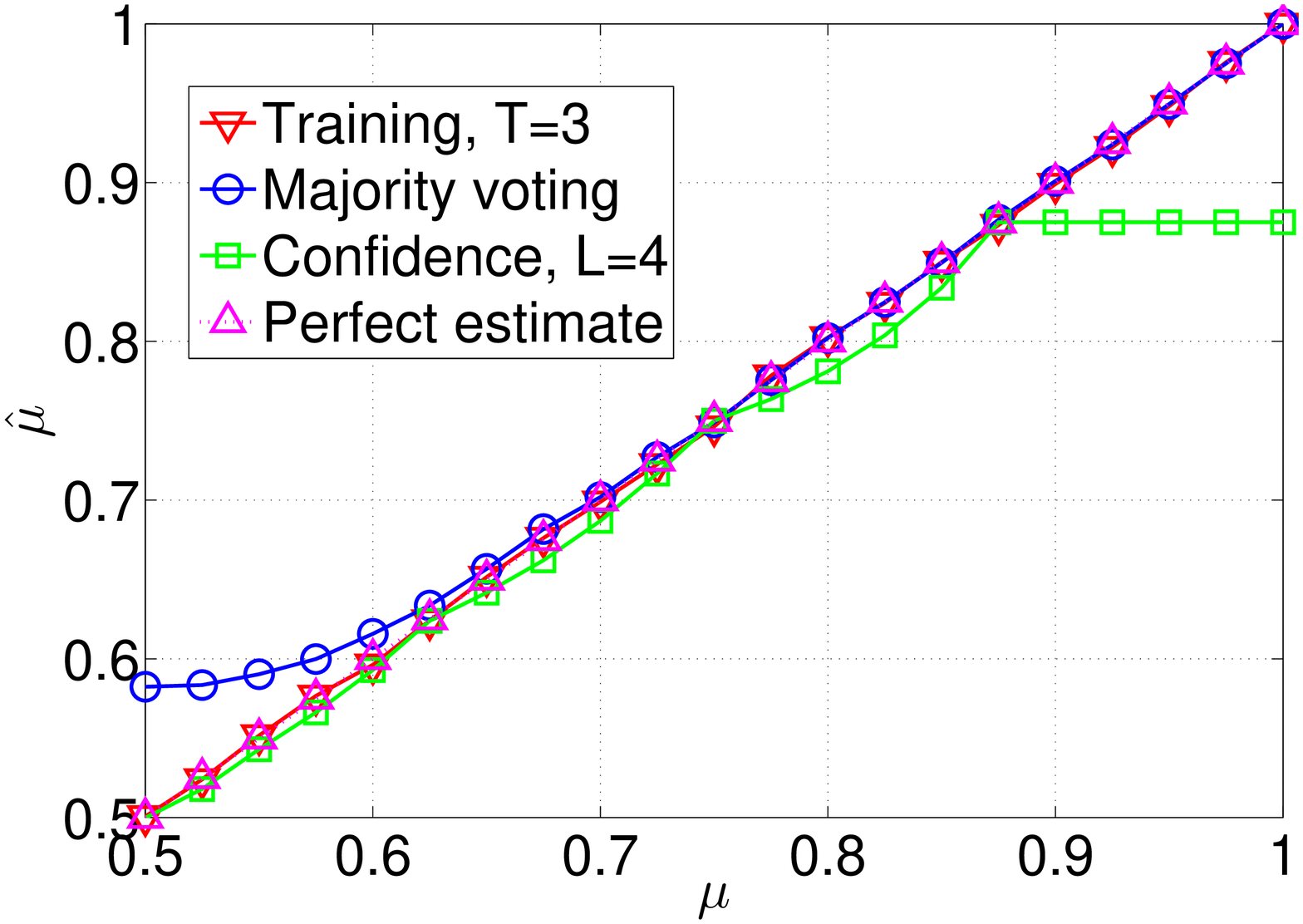}
		\subcaption{}
		\label{fig1}
	\end{minipage}%
	\begin{minipage}[t]{0.5\linewidth}
		\includegraphics[width=\linewidth]{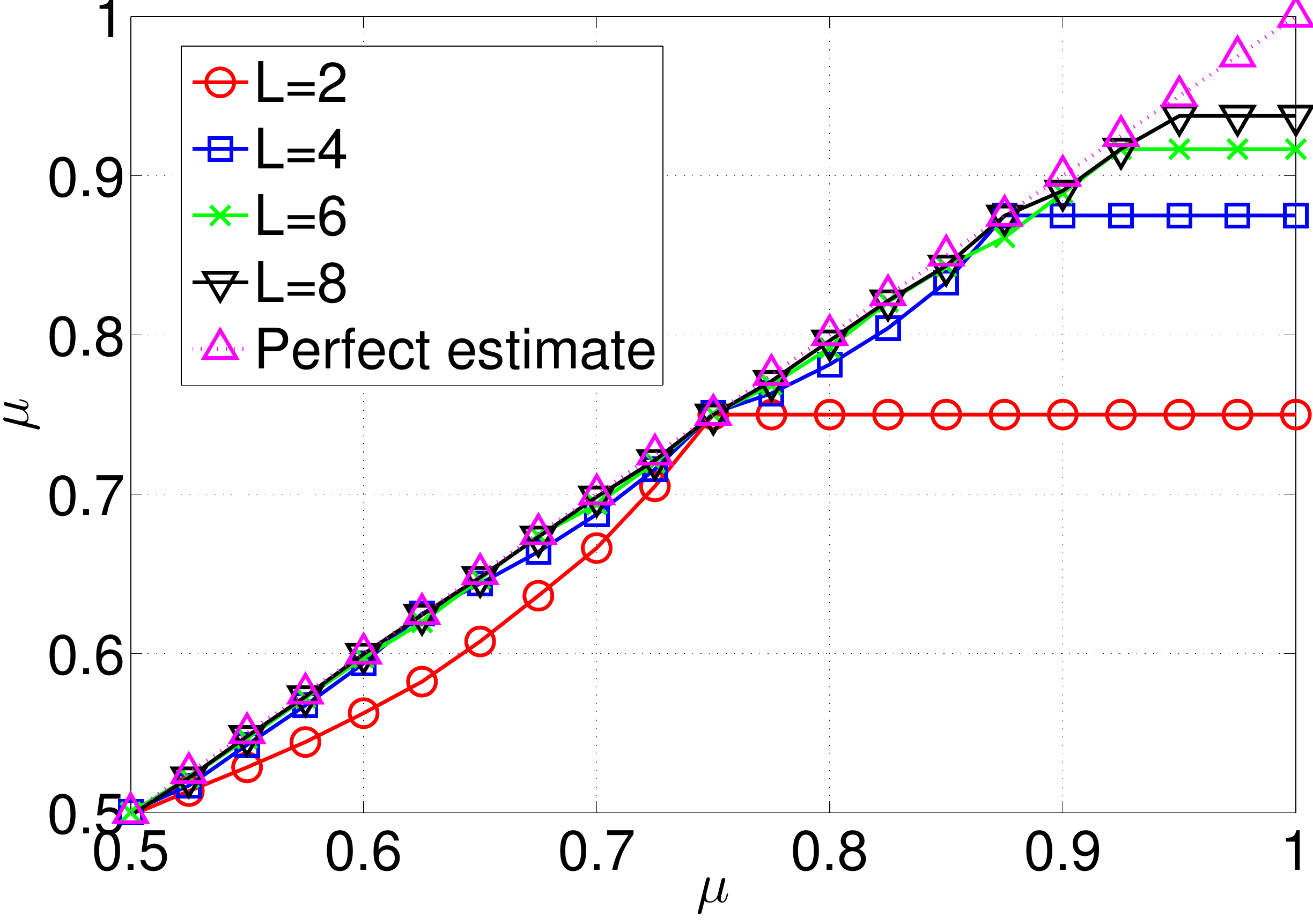}
		\subcaption{}
			\label{confidenceestimationcomparison}
	\end{minipage}
	\caption{ Estimation performance comparison. (a) Different methods. (b) Confidence-based method with different confidence levels.}
\end{figure}

In Fig. \ref{fig1}, it is observed that the training-based method has the best overall performance, which takes advantage of the gold standard questions. We can also see that the majority voting method has better performance as $\mu$ increases. This is because a larger $\mu$ means a better-quality crowd, which will lead to a more accurate result from majority voting, and consequently better estimation performance of $\mu$. When confidence is considered with $L=4$, we find that the overall estimation performance is not better than the other two methods because of quantization noise associated with confidence reporting in the estimation of $\mu$. It is also shown that the curve saturates and yields a fixed value of $\hat \mu=0.875$ when $\mu\ge 0.9$. This is because almost all the confidence levels submitted then are $l_{w,i}=4$ and the corresponding estimate result is exactly 0.875.

The estimation performance of the confidence-based method with multiple confidence levels is presented in Fig. \ref{confidenceestimationcomparison}. As is expected, a larger $L$ can help improve the estimation performance. However, it is seen that even though $L=8$, the corresponding performance is still not as good as that of the other two methods. Although we can expect estimation performance improvement as the maximum number of confidence levels $L$ increases, $L=8$ is pretty much the limit in practice due to the human inability to categorize beyond 7 levels. When the confidence-based estimation method is employed, the estimate value saturates at a certain fixed value when $\mu$ is large. Therefore, it can be concluded that the confidence-based estimation method does not provide good results.

\begin{figure}[h]
	\centering
	\includegraphics[width=3in]{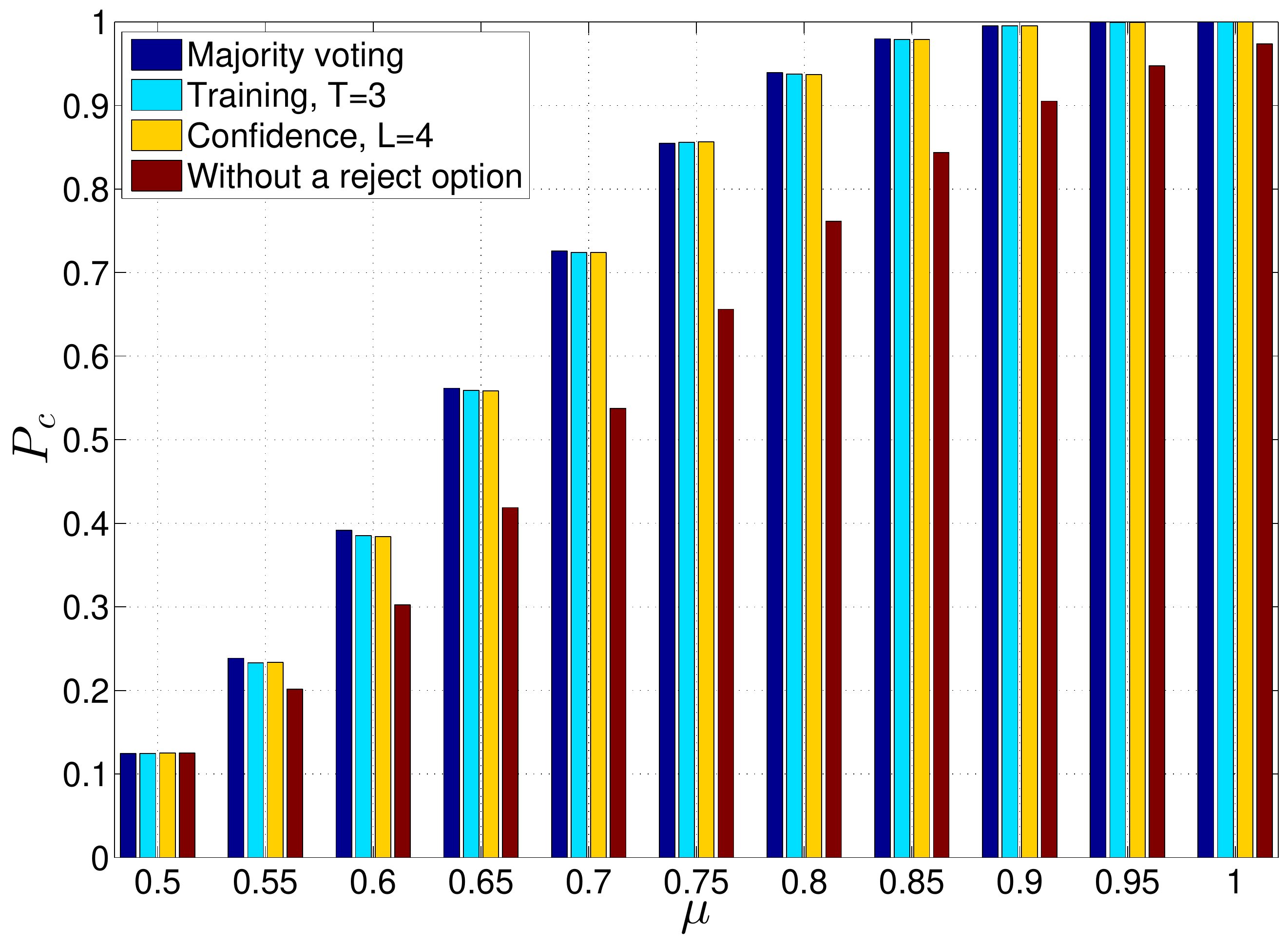} 
	\caption{Robustness of the proposed system and performance comparison with simple majority voting}
	\label{Pcwithestimate}
\end{figure}
Even though the three methods differ in performance in the estimation of $\mu$, we show in Fig. \ref{Pcwithestimate} the robustness of the proposed system. We observe from Fig. \ref{fig1} that the majority voting based method suffers from performance degradation in the low-$\mu$ regime, while the confidence based one suffers in the high-$\mu$ regime. However, when the value of $\mu$ is low, the workers are making random guesses even when they believe that they are able to respond with definitive answers. When the value of $\mu$ is large, almost all the definitive answers submitted are correct. Therefore, in those two situations, the performance degradation in the estimation of $\mu$ is negligible. From Fig. \ref{Pcwithestimate}, we see that system performance of the proposed system with estimation results from Fig. \ref{fig1} is almost the same as with the other three estimation methods, which significantly outperforms the system where simple majority voting is employed without a reject option. However, if a significant performance degradation in the estimation of $\mu$ occurs outside the two aforementioned regimes, overall classification performance loss is expected. For example, consider the case where $\mu$ is 0.8 while $\hat \mu$ is 0.5, and $N=5$, then $P_c=0.8$. However, the actual $P_c$ equals 0.89 when $\mu$ is estimated with an acceptable error.

\section{Conclusion}
We have studied a novel framework of crowdsourcing system for classification, where an individual worker has the reject option and can skip a microtask if he/she has no definitive answer, and gives definitive answers with quantized confidence. We presented an aggregation approach using a weighted majority voting rule, where each worker's response is assigned an optimized weight to maximize the crowd's classification performance. However, we showed that reporting of confidence by the crowd does not benefit classification performance. One is advised to adopt the reject option without confidence indication from the workers as it does not improve classification performance and may degrade performance in some cases.

\appendix
\section*{APPENDIX}
To solve problem \eqref{max}, we need $E_C[\mathbb{W}]$ and $E_O[\mathbb{W}]$. First, the $w$th worker can have weight contribution to $E_C[\mathbb{W}]$ only if all his/her definitive answers are correct. Thus, we have the average weight assigned to the correct element as
\begin{align}
&E_C[\mathbb{W}]\nonumber\\
& \!=\! E_{p,\rho, l}\!\left[\! {\sum\limits_{w = 1}^W\! {\sum\limits_{n = 0}^N\!\sum\limits_{N_n \in [N]}\!\prod \limits _{i\in N_n} \!\!\!{{W_w}(1\!-\!p_{w,i})\rho_{w,i}\!\!\!\!\!\!\!\prod \limits_{j\in [N]-N_n}\!\!\!\!\!\!\!\! p_{w,j}} } }\! |l_{w,i}\right]
\end{align}
where $[N]$ denotes $\{1,\ldots,N\}$ and $N_n\in [N]$ with cardinality $n$.	
Given a known $w$th worker, i.e., $p_{w,i}$ is known, we write
\begin{align}\label{5}
A_w(p_{w,i})= \sum\limits_{n = 0}^NE_{\rho,l}\left[{W_w}{\prod\limits_{i\in N_n}\rho_{w,i}} |l_{w,i} \right]P_{\lambda}(n),
\end{align}
where
$
P_{\lambda}(n)=\sum\limits_{N_n \in [N]}\prod \limits _{i\in N_n} {(1-p_{w,i})\prod \limits_{j\in [N]-N_n} p_{w,j}}.
$

Note that ${\sum\limits_{n = 0}^N {P_{\lambda}(n)}  }=1$, and then \eqref{5} is upper-bounded using Cauchy-Schwarz inequality as follows:
\begin{align}
&{A_w}(p_{w,i}) = \sum\limits_{n = 0}^N {E_{\rho,l} \left[ {{W_w}{\prod\limits_{i\in N_n}\rho_{w,i}}} |l_{w,i}\right]\sqrt {{P_\lambda }(n)} } \sqrt {{P_\lambda }(n)}  \nonumber \\ \label{8}
& \le \sqrt {\sum\limits_{n = 0}^N {{{{E_{\rho,l}^2}\left[ {{W_w}{\prod\limits_{i\in N_n}\rho_{w,i}}}|l_{w,i} \right]}}{P_{\lambda}(n)} } } \sqrt {\sum\limits_{n = 0}^N {P_{\lambda}(n)}  }.
\end{align}
Also note that equality holds in \eqref{8} only if 
\begin{align}
{{{E_{\rho,l}\left[ {{W_w}{\prod\limits_{i\in N_n}\rho_{w,i}}} |l_{w,i}\right]}\sqrt {P_{\lambda}(n)}  }}{{{}}} = \alpha_w(p_{w,i}) \sqrt {P_{\lambda}(n)} , 
\end{align}
where $\alpha_w$ is a positive quantity independent of $n$, which might be a function of $p_{w,i}$, and
\begin{align}\label{11}
{E_{\rho,l}\left[ {{W_w}{\prod\limits_{i\in N_n}\rho_{w,i}}} |l_{w,i}\right]}=\alpha_w(p_{w,i}).
\end{align}

Note that $\displaystyle{\int\limits_{p_{w,i}}\!\! {F_p ( {p_{w,i} = x} )dx}}=1$, and similarly we write
\begin{align}
& E_p[A_w(p_{w,i})]\le \int\limits_{p_{w,i}} {{\alpha _w(p_{w,i})}\Pr \left( {p_{w,i}= x} \right)dx}\nonumber\\
& \label{13} \le \sqrt {\int\limits_{p_{w,i}} {\alpha _w^2(p_{w,i})\Pr \left( {p_{w,i} = x} \right)dx} } \sqrt {\int\limits_{p_{w,i}} {\Pr \left( {p_{w,i} = x} \right)dx} }.
\end{align}
The equality \eqref{13} holds only if 
\begin{align}
\alpha_w(p_{w,i}) \sqrt {\Pr \left( {p_{w,i} = x}\right)}=\beta \sqrt {\Pr \left( {p_{w,i} = x}\right)},
\end{align}
WHERE $\beta$ is a positive constant independent of $p_{w,i}$, and we conclude that $\alpha_w$ is also a positive quantity independent of $p_{w,i}$.
Then from \eqref{11}, we have
$
{E_{\rho,l}\left[ {{W_w}{\prod\limits_{i\in N_n}\rho_{w,i}}}|l_{w,i} \right]}=\beta
$. 
Since ${\prod\limits_{i\in N_n}\rho_{w,i}}$ is the product of $n$ variables, its distribution is not known \textit{a priori}. A possible solution to weight assignment is a deterministic value given by $W_wE_{\rho,l}[{\prod\limits_{i\in N_n}\rho_{w,i}}|l_{w,i}]=\beta$ and, therefore, we can write the weight as 
\begin{align}
W_w=\frac{\beta}{E_{\rho,l}[{\prod\limits_{i\in N_n}\rho_{w,i}}|l_{w,i}]}=\frac{\beta}{\mu ^n}.
\end{align}

Then, we can express the crowd's average weight contribution to all the classes defined in \eqref{max} as
\begin{align}
{E_O}\left[ {{\mathbb{W}}} \right] &= \sum\limits_{w=1}^{W}E_{p,\rho,l}\left[ {\sum\limits_{n = 0}^N {\beta {\mu ^{ - n}}{2^{N - n}}{P_\lambda }\left( n \right)} } \right] \nonumber\\
&= \sum\limits_{w=1}^{W}\sum\limits_{n = 0}^N {\beta {\mu ^{ - n}}{2^{N - n}}\binom{N}{n}{{\left( {1 - m} \right)}^n}{m^{N - n}}}  \nonumber\\
&=W \beta {\left( {\frac{{1 - m}}{\mu } + 2m} \right)^N}=K.
\label{beta}
\end{align}
Thus, $\beta$ and the weight can be obtained accordingly.
Note that the weight derived above has a term that is common for every worker. Since the voting scheme is based on comparison, we can ignore this factor and have the normalized weight as $W_w={\mu}^{-n}$.

\begin{acks}
This work was supported in part by the Army Research Office under Grant W911NF-14-1-0339 and in part by the National Science Foundation under Grant ENG-1609916.
\end{acks}

\balance
\bibliographystyle{ACM-Reference-Format}

\bibliography{ref_Lqw} 

\end{document}